\documentclass[11pt]{article}

\usepackage{fullpage}
\usepackage{amsmath}
\usepackage{amsfonts}
\usepackage{amssymb}
\usepackage{amsthm}
\usepackage{bm}
\usepackage{graphicx}
\usepackage[dvipsnames]{xcolor}
\usepackage{hyperref}
\usepackage[nottoc]{tocbibind}
\allowdisplaybreaks

\newtheorem{theorem}{Theorem}
\newtheorem{lemma}{Lemma}
\newtheorem{proposition}{Proposition}

\theoremstyle{definition}
\newtheorem{definition}{Definition}
\newtheorem{mechanism}{Mechanism}

\theoremstyle{remark}

\title{Nearly Complete Characterization of 2-Agent Deterministic Strategyproof Mechanisms
for Single Facility Location in $L_p$ Space}

\author{Jianan Lin\\[0.6em]
\small School of Computer Science, Fudan University,\\
\small Shanghai 201203, China\\
\small \texttt{jnlin16@fudan.edu.cn}}
\date{}

\begin{document}
\maketitle
\begingroup
\renewcommand{\thefootnote}{}
\footnotetext{I thank my advisors Pinyan Lu and Hu Fu for their advice on this problem. This version was extensively revised with assistance from GPT Astra-6 to repair flawed proofs and improve the language of the COCOA~2020 version.}
\endgroup

\begin{abstract}
We consider the problem of locating a single facility for two agents in $L_p$ space ($1<p<\infty$) and study the characterization of deterministic strategyproof mechanisms. The cost of an agent is her $L_p$ distance to the facility. A mechanism is strategyproof if no agent can reduce her cost by misreporting her private location.

We show that, for two agents, the output of a deterministic, unanimous, translation-invariant strategyproof mechanism must satisfy a set of necessary equations, and that such mechanisms are continuous and scalable. In one dimension, the output must be one agent's location, for any number of agents. Scalability also holds for any number of agents.

In higher dimensions, the output need not coincide with an agent's location. When $p=2$, the necessary equations require the output to lie on the sphere with the segment between the two agents as a diameter. For more agents, this condition does not extend directly to dimensions greater than two: the coordinatewise median mechanism provides a counterexample.

In Euclidean space of dimension at least two, a two-agent mechanism in this class is rotation-invariant if and only if it is dictatorial. In the Euclidean plane, we characterize the anonymous mechanisms as three explicit families. Finally, our argument gives a short proof of the tight approximation ratio of 2 for deterministic strategyproof mechanisms under the maximum-cost objective in any dimension.
\end{abstract}

\section{Introduction}\label{sec:introduction}

Facility location provides a basic setting for studying mechanism design without monetary transfers. A facility must be selected from reports of agents' preferred locations, and each agent wishes to minimize her distance to the selected point. Since the locations are private information, optimizing an aggregate objective does not by itself ensure truthful reporting. Strategyproofness requires that no agent can reduce her distance by changing her report, regardless of the other agents' reports.

Two complementary questions arise in this setting. The first asks which mechanisms are strategyproof under a given preference domain and additional requirements such as unanimity or anonymity. Moulin's characterization on single-peaked domains relates strategyproofness to generalized median rules, which may include fixed phantom reports rather than simply selecting the median of the actual reports \cite{Moulin1980}. Border and Jordan study multidimensional domains of separable star-shaped preferences and obtain coordinatewise decompositions under their assumptions \cite{Border1983}. The preference domain matters: requiring strategyproofness for a rich family of preferences is different from requiring it only for distances induced by a fixed norm.

The second question asks how well strategyproof mechanisms approximate a social objective. Procaccia and Tennenholtz initiate the study of approximate mechanism design without money, including single-facility location on the line under social cost and maximum cost \cite{Procaccia2009}. Feigenbaum et al.\ study the $L_p$ norm of the vector of agents' costs on the line and establish optimal deterministic approximation guarantees, as well as randomized results for two agents \cite{LpNorm}. Their objective should be distinguished from our use of an $L_p$ norm to measure distances in the underlying location space.

We consider a fixed $L_p$ distance in finite-dimensional space and focus on deterministic, unanimous, translation-invariant mechanisms. Translation invariance means that moving every report by the same vector moves the output by that vector, without privileging a particular origin. We ask what restrictions strategyproofness imposes on the output, and how additional symmetry requirements constrain the mechanism. Most of the analysis concerns two agents; several preliminary properties hold for arbitrary numbers of agents. We do not impose Pareto efficiency or group-strategyproofness.

\paragraph{Our Contributions.}
Our work is motivated by the two-agent randomized results on the line in \cite{LpNorm}, which lead us to study deterministic mechanisms in multidimensional spaces. We show that in one dimension, for any number of agents, the output of a deterministic unanimous translation-invariant strategyproof mechanism must be one agent's location. In multidimensional $L_p$ space ($1<p<\infty$), the output of such a mechanism with two agents must satisfy a set of necessary equations. In particular, in Euclidean space, the two reports $A,B$ and the output $W$ satisfy $(W-A)\cdot(W-B)=0$. We also establish continuity for two agents and scalability for any number of agents. We then give complete characterizations in two more specific settings, both restricted to two agents: in Euclidean dimension at least two, rotation invariance is equivalent to dictatorship; in the Euclidean plane, anonymity is equivalent to belonging to one of the three mechanism families in Section~\ref{sec:symmetry}. Finally, we show that the coordinatewise median is a counterexample to the direct extension of the two-agent orthogonality condition to arbitrary numbers of agents in dimensions greater than two, and give a short proof of the tight maximum-cost approximation ratio of 2 in any dimension.

\paragraph{Additional Related Work.}
Generalized median mechanisms have been studied beyond the original one-dimensional setting. Barber\`a et al.\ develop characterizations based on generalized single-peaked domains and voting by committees \cite{Barber1993}, and study strategyproof voting with compact ranges \cite{Barbera1998}. These results involve preference-domain and range assumptions that differ from the fixed-distance setting considered here. Tang et al.\ study group-strategyproof facility location mechanisms in strictly convex spaces, treating both deterministic and randomized mechanisms \cite{Tang2018}. Our analysis instead concerns deviations by individual agents.

Other work changes the geometry, the number of facilities, or the agents' preferences. Alon et al.\ investigate strategyproof approximation of maximum cost on networks \cite{Alon2010}. Fotakis and Tzamos study winner-imposing mechanisms for multiple facilities, where an agent allocated a facility is required to connect to it \cite{Fotakis2013}. Heterogeneous facilities and fractional preferences provide further variants \cite{serafino2014heterogeneous,fong2018facility}. These models motivate studying how geometric and preference assumptions affect the structure of strategyproof mechanisms, but are not covered by our characterizations.

\section{Preliminaries}\label{sec:preliminaries}

\paragraph{Model.}
There are $n\ge2$ agents, indexed by $N=\{1,\ldots,n\}$, with locations in $\mathbb R^m$, where $m\ge1$ is finite. Unless stated otherwise, distances are induced by an $L_p$ norm with $1<p<\infty$:
\[
d(A,B)=\|A-B\|_p=\bigg(\sum_{j=1}^m|A_j-B_j|^p\bigg)^{\frac{1}{p}}.
\]
We omit the subscript when the norm is clear. The case $p=2$ is Euclidean space, with the usual inner product denoted by a dot. A profile is $\bm A=(A_1,\ldots,A_n)$. For a fixed agent $i$, the notation $\bm A_{-i}$ denotes the reports of all other agents; $f(A_i',\bm A_{-i})$ means that only the $i$th report has been replaced. For two agents we write the reports as $A,B$.

A deterministic mechanism is a map $f:(\mathbb R^m)^n\to\mathbb R^m$. Each agent's cost is the distance from her true location to the output. There are no payments. We use the following requirements.

\begin{definition}[Strategyproofness]
A mechanism is strategyproof if, for every profile $\bm A$, every agent $i$, and every alternative report $A_i'$, it satisfies $\|A_i-f(\bm A)\|\le\|A_i-f(A_i',\bm A_{-i})\|$.
\end{definition}

\begin{definition}[Unanimity and Dictatorship]
A mechanism is unanimous if $f(C,\ldots,C)=C$ for every $C\in\mathbb R^m$. It is dictatorial if there is a fixed agent $i$ such that $f(\bm A)=A_i$ for every profile.
\end{definition}

\begin{definition}[Anonymity]
A mechanism is anonymous if $f(A_{\sigma(1)},\ldots,A_{\sigma(n)})=f(A_1,\ldots,A_n)$ for every profile and every permutation $\sigma$ of $N$.
\end{definition}

\begin{definition}[Translation Invariance and Scalability]
A mechanism is translation-invariant if $f(A_1+t,\ldots,A_n+t)=f(\bm A)+t$ for every $t\in\mathbb R^m$. It is scalable if $f(kA_1,\ldots,kA_n)=kf(\bm A)$ for every $k>0$.
\end{definition}

\begin{definition}[Rotation Invariance]
In Euclidean space of dimension $m\ge2$, a mechanism is rotation-invariant if $f(QA_1,\ldots,QA_n)=Qf(\bm A)$ for every orthogonal matrix $Q$ with determinant 1.
\end{definition}

Together with translation invariance, this definition also implies equivariance under rotations about any point. Reflections are not required. Scalability and continuity are conclusions of our analysis, not assumptions on the mechanisms.

\paragraph{Geometric and Approximation Notation.}
The segment joining $A$ and $B$ is $[A,B]=\{(1-t)A+tB:0\le t\le1\}$. For $1<p<\infty$, the norm is strictly convex: equality in $\|x+y\|\le\|x\|+\|y\|$ holds only when one vector is zero or the two nonzero vectors have the same direction. This property is used in the segment argument below; the general analysis does not cover $p=1$ or $p=\infty$.

The maximum cost at an output $W$ is $\operatorname{MC}(\bm A,W)=\max_i\|A_i-W\|$, and its optimum is $\operatorname{OPT}(\bm A)=\min_W\operatorname{MC}(\bm A,W)$. A mechanism has approximation ratio at most $r$ if $\operatorname{MC}(\bm A,f(\bm A))\le r\operatorname{OPT}(\bm A)$ for every profile. This formulation also requires optimal cost whenever the optimum is zero. Our maximum-cost argument applies to all $L_p$ norms, including the two endpoints.

\section{Necessary Conditions and Regularity}\label{sec:necessary}

We begin with properties of strategyproof mechanisms on the line and a segment argument in strictly convex spaces. We then derive the two-agent Euclidean condition and its $L_p$ counterpart. Throughout this section, assumptions needed for each result are stated explicitly.

\subsection{One-Dimensional Mechanisms}

\begin{lemma}[Continuity of an Agent's Cost]\label{lem:cost-continuity}
For a deterministic strategyproof mechanism and fixed reports $\bm A_{-i}$, the function $u(A_i)=\|A_i-f(A_i,\bm A_{-i})\|$ satisfies $|u(A_i)-u(A_i')|\le\|A_i-A_i'\|$.
\end{lemma}

\begin{proof}
Strategyproofness and the triangle inequality give
\[
u(A_i)\le\|A_i-f(A_i',\bm A_{-i})\|
\le\|A_i-A_i'\|+u(A_i').
\]
Interchanging the two reports proves the claim. This argument works for any number of agents and in any normed space; it concerns the agent's cost, not the output map.
\end{proof}

\begin{theorem}\label{thm:line}
On the real line, every deterministic unanimous translation-invariant strategyproof mechanism outputs one of the reported locations.
\end{theorem}

\begin{proof}
First, strategyproofness implies that the output is nondecreasing in each agent's report when the other reports are fixed. To see this, let $a<a'$ be two reports and let $w,w'$ be their respective outputs. Squaring the two strategyproofness inequalities gives
\[
(a-w)^2\le(a-w')^2,\qquad (a'-w')^2\le(a'-w)^2.
\]
Adding and simplifying yields $(a'-a)(w'-w)\ge0$, so $w'\ge w$.

Now suppose that $W=f(\bm A)$ is different from every $A_i$. Choose $0<\epsilon<\min_i|A_i-W|$. Increase the reports by $\epsilon$, one at a time, and denote the successive outputs by $W_0,\ldots,W_n$. Monotonicity and translation invariance imply
\[
W=W_0\le W_1\le\cdots\le W_n=f(\bm A+\epsilon)=W+\epsilon.
\]
At step $i$, both outputs $W_{i-1},W_i$ are on the same side of both reports $A_i,A_i+\epsilon$. If $A_i>W$, then both outputs are below $A_i$, and $W_i>W_{i-1}$ would give the agent at $A_i$ a profitable deviation. If $A_i<W$, then both outputs are above $A_i+\epsilon$, and $W_i>W_{i-1}$ would give the agent at $A_i+\epsilon$ a profitable deviation back to $A_i$. Thus $W_i=W_{i-1}$ in every step, contradicting $W_n=W+\epsilon$.
\end{proof}

In one dimension all $L_p$ norms agree with absolute distance. The conclusion therefore does not depend on $p$. In higher dimensions an output need not coincide with a report.

\subsection{The Euclidean Condition}

The following segment property is the main geometric tool. Although we first use it in Euclidean space, its proof applies to all strictly convex $L_p$ norms.

\begin{lemma}\label{lem:segment}
Let $f$ be a deterministic strategyproof mechanism in $L_p$ space, where $1<p<\infty$, and let $W=f(\bm A)$. If $A_i'\in[A_i,W]$, then $f(A_i',\bm A_{-i})=W$.
\end{lemma}

\begin{proof}
The case $A_i'=A_i$ is immediate. Otherwise, write $W'=f(A_i',\bm A_{-i})$. Strategyproofness at the two reports, together with the triangle inequality, gives
\[
\|A_i-W\|\le\|A_i-W'\|\le\|A_i-A_i'\|+\|A_i'-W'\|
\le\|A_i-A_i'\|+\|A_i'-W\|=\|A_i-W\|.
\]
Hence every inequality is an equality. If $A_i'=W$, the second strategyproofness inequality gives $W'=W$ directly. Otherwise, strict convexity implies that $W'$ lies on the ray from $A_i$ through $A_i'$, beyond $A_i'$. Since $\|A_i'-W'\|=\|A_i'-W\|$, its position on that ray is exactly $W$.
\end{proof}

This argument uses only strategyproofness and strict convexity; neither unanimity nor translation invariance is needed for this lemma.

\begin{lemma}\label{lem:collinear}
For two agents in Euclidean space, let $f$ be deterministic, unanimous, translation-invariant, and strategyproof. If $A\ne B$ and $f(A,B)$ lies on the line through $A,B$, then the output is $A$ or $B$.
\end{lemma}

\begin{proof}
Let $W=f(A,B)$. If $W$ lies outside the segment $AB$, move both agents the same sufficiently small distance toward $W$. Lemma~\ref{lem:segment} keeps the output fixed at $W$, whereas translation invariance requires the output to move by that same nonzero displacement, a contradiction.

Suppose instead that $W$ is strictly between $A$ and $B$. Name the closer endpoint $A$, so that $a=d(A,W)\le b=d(W,B)$; this merely names the endpoints and does not assume anonymity. By Lemma~\ref{lem:segment}, moving the agent at $A$ to $W$ gives output $W$. Set $B'=B+(W-A)$. Translation invariance applied to the original profile gives $f(W,B')=W+(W-A)$. Therefore, with the other agent fixed at $W$, the agent at $B$ can change her report to $B'$ and reduce her cost from $b$ to $b-a<b$. This contradicts strategyproofness.
\end{proof}

\begin{lemma}\label{lem:acute}
For two agents in Euclidean space, a deterministic unanimous translation-invariant strategyproof mechanism cannot have an output $W\notin\{A,B\}$ with $0<\angle AWB<\frac{\pi}{2}$.
\end{lemma}

\begin{proof}
Suppose such an output exists and put $v=W-A$. For $t\in[0,1]$, define $A_t=A+tv$, $B_t=B+tv$, and $W_t=W+tv$. Translation invariance gives $f(A_t,B_t)=W_t$. The point $W$ lies on $[A_t,W_t]$, so the segment property gives $f(W,B_t)=W_t$.

Consider the second agent truly located at $B_1$. Reporting $B_{1-\epsilon}$ instead of $B_1$, for $0<\epsilon<1$, changes the output from $W_1$ to $W_{1-\epsilon}$. The difference of squared costs is
\[
\|B_1-W_{1-\epsilon}\|^2-\|B_1-W_1\|^2
=2\epsilon(B-W)\cdot v+\epsilon^2\|v\|^2.
\]
The acute-angle assumption gives $(B-W)\cdot v=-(B-W)\cdot(A-W)<0$. The displayed difference is therefore negative for sufficiently small $\epsilon>0$, contradicting strategyproofness.
\end{proof}

\begin{lemma}\label{lem:obtuse}
Under the assumptions of Lemma~\ref{lem:acute}, an output $W\notin\{A,B\}$ cannot satisfy $\frac{\pi}{2}<\angle AWB<\pi$.
\end{lemma}

\begin{proof}
Use the same points $A_t,B_t,W_t$ as in the preceding proof, so that $f(W,B_t)=W_t$. This time let the second agent's true location be $B=B_0$, and consider a report $B_\epsilon$. The difference of squared costs is
\[
\|B-W_\epsilon\|^2-\|B-W\|^2
=-2\epsilon(B-W)\cdot v+\epsilon^2\|v\|^2.
\]
Since the angle is obtuse, $(B-W)\cdot v>0$. A sufficiently small positive $\epsilon$ again gives a profitable deviation.
\end{proof}

\begin{theorem}\label{thm:sphere}
For two agents in Euclidean space, the output $W=f(A,B)$ of a deterministic unanimous translation-invariant strategyproof mechanism satisfies
\[(W-A)\cdot(W-B) = 0.\]
Thus, for distinct reports, $W$ lies on the sphere with diameter $AB$; for coincident reports, $W=A=B$.
\end{theorem}

\begin{proof}
If the reports coincide, unanimity settles the claim. For distinct reports, the collinear case is covered by Lemma~\ref{lem:collinear}. Otherwise Lemmas~\ref{lem:acute} and~\ref{lem:obtuse} exclude acute and obtuse angles, leaving a right angle. Outputs equal to either report also satisfy the equation.
\end{proof}

The sphere condition constrains two-agent outputs without requiring anonymity or rotation invariance. Section~\ref{sec:discussion} examines why its direct extension to arbitrary numbers of agents fails.

\subsection{Necessary Equations in $L_p$ Space}

Fix $1<p<\infty$. By translation invariance, we can place the midpoint of the two reports at the origin and write $A=-x$, $B=x$, and $W=y$. The Euclidean inner-product condition then has the following analogue.

\begin{theorem}\label{thm:lp-equations}
For two agents in $L_p$ space with $1<p<\infty$, let $f$ be deterministic, unanimous, translation-invariant, and strategyproof. Write $A=-x$, $B=x$, and $f(A,B)=y$. Then
\[
\left\{
\begin{aligned}
\sum_{i=1}^{m}{(x_i+y_i)\cdot (x_i-y_i) \cdot |x_i-y_i|^{p-2}} = 0 \\
\sum_{i=1}^{m}{(x_i-y_i)\cdot (x_i+y_i) \cdot |x_i+y_i|^{p-2}} = 0
\end{aligned}
\right.
\]
Here $t|t|^{p-2}$ is interpreted as zero when $t=0$.
\end{theorem}

\begin{proof}
Write $v=W-A$ and, for $t\in[0,1]$, set $A_t=A+tv$, $B_t=B+tv$, and $W_t=W+tv$. Translation invariance gives $f(A_t,B_t)=W_t$. Since $W$ lies on the segment $A_tW_t$, Lemma~\ref{lem:segment} further gives $f(W,B_t)=W_t$ for every such $t$.

Fix $t=\frac{1}{2}$. With the first report fixed at $W$, strategyproofness for the second agent at $B_t$ implies $\|B_t-W_t\|\le\|B_t-W_s\|$ for every $s\in[0,1]$. Writing $\alpha=s-t$, we find that the differentiable function
\[
G(\alpha)=\frac{1}{p}\sum_{i=1}^m|W_i-B_i+\alpha(W_i-A_i)|^p
\]
has a minimum at $\alpha=0$ on $\big[-\frac{1}{2},\frac{1}{2}\big]$. Therefore
\[
0=G'(0)=\sum_{i=1}^m(W_i-A_i)(W_i-B_i)|W_i-B_i|^{p-2}.
\]
The derivative formula is valid for all $p>1$, including when $W_i=B_i$, with the convention in the statement. Substituting $A=-x$, $B=x$, and $W=y$ yields the first equation after multiplying by $-1$. Interchanging the roles of the two agents yields the second equation. The argument also covers $W=A$ or $W=B$, without dividing by either distance.
\end{proof}

When $p=2$, interpreting the factors $|x_i\pm y_i|^{p-2}$ as 1, the two equations coincide and describe the sphere with diameter $AB$. For $p\ne2$, these necessary equations need not have finitely many solutions. For example, when $m=3$, $p=4$, and $x=(1,0,0)$, every $y=(0,a,b)$ with $a^4+b^4=1$ satisfies both equations.

\begin{theorem}\label{thm:scalability}
In any $L_p$ space with $1<p<\infty$ and $n\ge2$ agents, a deterministic translation-invariant strategyproof mechanism is scalable.
\end{theorem}

\begin{proof}
Let $W=f(A_1,\ldots,A_n)$ and fix $0<k<1$. For every $i$, the point $A_i'=kA_i+(1-k)W$ lies on the segment $A_iW$. Applying the argument of Lemma~\ref{lem:segment} to the agents in succession gives $f(A_1',\ldots,A_n')=W$. Translating every report by $-(1-k)W$ therefore gives
\[
f(kA_1,\ldots,kA_n)=f(A_1',\ldots,A_n')-(1-k)W=kW.
\]
For $k>1$, apply the result just proved with scale $\frac{1}{k}$ to the profile $(kA_1,\ldots,kA_n)$. The case $k=1$ is immediate. Neither the segment argument nor this proof requires unanimity.
\end{proof}

\begin{proposition}\label{prop:output-continuity}
For two agents in $L_p$ space with $1<p<\infty$, every deterministic unanimous translation-invariant strategyproof mechanism has a continuous output map.
\end{proposition}

\begin{proof}
Write $g(x)=f(0,x)$ and $h(x)=x-g(x)$, and let $S$ and $T$ be their respective ranges. Translation invariance gives $f(A,B)=A+g(B-A)$. Unanimity gives $g(0)=h(0)=0$, and Theorem~\ref{thm:scalability} implies that both ranges are closed under multiplication by nonnegative scalars.

With the first report fixed at zero, strategyproofness says that $g(x)$ is a closest point in $S$ to $x$. With the second report fixed at zero, the output at the first report $-x$ is $f(-x,0)=-h(x)$, and the attainable outputs are precisely $-T$. Strategyproofness for this agent therefore says that $h(x)$ is a closest point in $T$ to $x$. In particular, $g(y)=y$ for $y\in S$ and $h(z)=z$ for $z\in T$. The latter identity implies $g(z)=0$.

We show that $S$ is closed and convex. For a vector $z$, define $J_p(z)$ coordinatewise by $(J_p(z))_j=z_j|z_j|^{p-2}$, with value zero when $z_j=0$. Fix $z\in T$ and $y\in S$. Since $g(z)=0$ and $ty\in S$ for all $t>0$, the closest-point property gives $\|z\|_p^p\le\|z-ty\|_p^p$. Taking the right derivative at $t=0$ gives $J_p(z)\cdot y\le0$. Thus
\[
S\subseteq H:=\{x\in\mathbb R^m:J_p(z)\cdot x\le0
                       \text{ for every }z\in T\}.
\]
Conversely, take $x\in H$ and set $u=g(x)$ and $v=h(x)=x-u$. The necessary equations in Theorem~\ref{thm:lp-equations}, applied to the reports $0,x$ and output $u$, imply $J_p(v)\cdot u=0$. This is the same equation after translating the midpoint to the origin; it also holds when either vector is zero. Since $v\in T$ and $x\in H$, we obtain
\[
0\ge J_p(v)\cdot x=J_p(v)\cdot(u+v)=\|v\|_p^p.
\]
Hence $v=0$ and $x=u\in S$, proving $S=H$. As an intersection of closed half-spaces, $S$ is closed and convex.

The closest point in $S$ to any $x$ is unique. Indeed, if two distinct points both minimized the distance, their midpoint would belong to $S$ and strict convexity would give a strictly smaller distance. To prove continuity, let $x_k\to x$. Since $0\in S$, the closest-point property yields $\|x_k-g(x_k)\|_p\le\|x_k\|_p$, so the sequence $g(x_k)$ is bounded. Any convergent subsequence has a limit $y\in S$, and passing to the limit in $\|x_k-g(x_k)\|_p\le\|x_k-z\|_p$ for each $z\in S$ shows that $y$ is a closest point to $x$. Uniqueness implies $y=g(x)$. Every convergent subsequence has this same limit, so boundedness gives $g(x_k)\to g(x)$. Finally, the identity $f(A,B)=A+g(B-A)$ proves joint continuity in both reports.
\end{proof}

\section{Characterizations under Symmetry}\label{sec:symmetry}

We now impose additional symmetry requirements on the two-agent mechanisms. Rotation invariance gives dictatorship in Euclidean dimensions $m\ge2$, whereas anonymity gives coordinatewise rules in the Euclidean plane.

\subsection{Rotation-Invariant Mechanisms}

\begin{theorem}\label{thm:rotation}
Let $f$ be a deterministic unanimous translation-invariant strategyproof mechanism for two agents in Euclidean space of dimension $m\ge2$. Then $f$ is rotation-invariant if and only if it is dictatorial.
\end{theorem}

\begin{proof}
A dictatorial mechanism is clearly rotation-invariant. Conversely, suppose that $f$ is rotation-invariant. Together with translation invariance and Theorem~\ref{thm:scalability}, this means that $f(T(A),T(B))=T(f(A,B))$ for every map $T(z)=kQz+t$, where $k>0$ and $Q$ is a rotation.

Suppose that some distinct reports $A,B$ have output $W\notin\{A,B\}$. Lemma~\ref{lem:segment} gives $f(W,B)=W$. In dimension at least two, there is a map $T$ of the above form sending the ordered pair $(A,B)$ to $(W,B)$. Equivariance would then give $T(W)=f(W,B)=W=T(A)$. Since $T$ is injective, this contradicts $W\ne A$.

Thus the output is always an endpoint. Fix one ordered pair of distinct reports. Every other ordered pair of distinct reports is its image under such a map $T$, so the identity of the selected agent is the same for every pair. Unanimity handles coincident reports, and hence $f$ is dictatorial.
\end{proof}

\subsection{Anonymous Mechanisms}

We describe three families in the original planar coordinates. Write $A=(x_A,y_A)$ and $B=(x_B,y_B)$, and set the output equal to $A$ whenever $A=B$. Sorting the two reports in the definitions below does not depend on their identities.

\begin{mechanism}[C1]\label{mech:c1}
Fix $u,v\in\{0,1\}$. Choose $x_W=\max\{x_A,x_B\}$ if $u=1$ and $x_W=\min\{x_A,x_B\}$ if $u=0$. Independently, choose $y_W=\max\{y_A,y_B\}$ if $v=1$ and $y_W=\min\{y_A,y_B\}$ if $v=0$.
\end{mechanism}

\begin{mechanism}[C2]\label{mech:c2}
Fix a real parameter $u\ne0$ and order the reports so that $x_A\le x_B$. If $x_A=x_B$, set $x_W=x_A$ and choose $y_W=\max\{y_A,y_B\}$ for $u>0$, or $y_W=\min\{y_A,y_B\}$ for $u<0$.

Otherwise, let $R=\frac{y_B-y_A}{x_B-x_A}$. If $R$ lies between $u$ and $-\frac{1}{u}$, including the endpoints, choose the intersection of the lines $y=u(x-x_A)+y_A$ and $y=-\frac{1}{u}(x-x_B)+y_B$.

For $u>0$, choose $B$ when $R>u$ and $A$ when $R<-\frac{1}{u}$. For $u<0$, choose $A$ when $R>-\frac{1}{u}$ and $B$ when $R<u$.
\end{mechanism}

\begin{mechanism}[C3]\label{mech:c3}
Fix a real parameter $v\ne0$ and order the reports so that $y_A\le y_B$. If $y_A=y_B$, set $y_W=y_A$ and choose $x_W=\max\{x_A,x_B\}$ for $v>0$, or $x_W=\min\{x_A,x_B\}$ for $v<0$.

Otherwise, let $S=\frac{x_B-x_A}{y_B-y_A}$. If $S$ lies between $v$ and $-\frac{1}{v}$, including the endpoints, choose the intersection of the lines $x=v(y-y_A)+x_A$ and $x=-\frac{1}{v}(y-y_B)+x_B$.

For $v>0$, choose $B$ when $S>v$ and $A$ when $S<-\frac{1}{v}$. For $v<0$, choose $A$ when $S>-\frac{1}{v}$ and $B$ when $S<v$.
\end{mechanism}

\begin{theorem}\label{thm:anonymous}
For two agents in the Euclidean plane, a mechanism is deterministic, unanimous, translation-invariant, anonymous, and strategyproof if and only if it belongs to one of the families in Mechanisms~\ref{mech:c1}--\ref{mech:c3}.
\end{theorem}

\begin{proof}
First, we prove sufficiency. Each mechanism is visibly deterministic, unanimous, translation-invariant, and anonymous. We verify strategyproofness by expressing it as a coordinatewise minimum or maximum in an orthonormal basis.

For Mechanism~\ref{mech:c1}, this is its definition in the standard basis. For Mechanism~\ref{mech:c2}, use the orthonormal vectors $e=\frac{(1,u)}{\sqrt{1+u^2}}$ and $h=\frac{(-u,1)}{\sqrt{1+u^2}}$. Its $e$-coordinate is the maximum of the two reported $e$-coordinates. Its $h$-coordinate is their maximum if $u>0$, and their minimum if $u<0$. Indeed, in the intersection case the first line fixes $h\cdot W=h\cdot A$ and the second fixes $e\cdot W=e\cdot B$. Comparing $e\cdot(B-A)$ and $h\cdot(B-A)$ gives exactly the endpoint branches in the definition, including the case $x_A=x_B$. Mechanism~\ref{mech:c3} follows by exchanging the two original coordinates.

To see why each of these rules is strategyproof, fix the other agent's report. In any coordinate where the rule takes the maximum, every attainable output coordinate lies in the half-line above the other report, and truthful reporting selects the closest point in that half-line to the agent's true coordinate. The same statement holds with the lower half-line for a minimum. Consequently, no misreport can reduce the absolute error in any of the transformed coordinates. Because the basis is orthonormal, the squared Euclidean cost is the sum of the squared coordinate errors, so no misreport can reduce the total cost.

We now prove necessity. Rather than dividing the possible outputs on the two coordinate axes into cases, we determine the set of outputs attainable when one agent is fixed at the origin.

\textbf{Step 1: Describe the attainable outputs.}

Write $g(x)=f(0,x)$ and $K=\{g(x):x\in\mathbb R^2\}$. Translation invariance and anonymity give
\[
f(A,B)=A+g(B-A),\qquad g(x)=x+g(-x).
\]
By unanimity, $0\in K$. By Theorem~\ref{thm:scalability}, $g(tx)=tg(x)$ for $t>0$, so $K$ is closed under multiplication by nonnegative scalars. Moreover, $g(z)=z$ for every $z\in K$: with the first report fixed at zero, an agent truly located at $z$ can obtain cost zero by using a report that produces $z$, and strategyproofness requires truthful reporting to give the same cost. Consequently, the preceding identity gives $g(-z)=0$ for every $z\in K$.

For any $x$, strategyproofness also implies that $g(x)$ is a closest point in $K$ to $x$. Indeed, every $z\in K$ can be obtained by some report, so $\|x-g(x)\|\le\|x-z\|$. We have not yet assumed that $K$ is convex or that a closest point is unique.

\textbf{Step 2: Show that $K$ equals its dual cone.}

Define $K^*=\{x\in\mathbb R^2:x\cdot z\ge0\text{ for every }z\in K\}$. We first show that any two vectors in $K$ have a nonnegative inner product. Fix $y,z\in K$. Since $g(-y)=0$ and $tz\in K$ for every $t>0$, strategyproofness at the true report $-y$ gives
\[
\|y\|^2\le\|y+tz\|^2
=\|y\|^2+2t\,y\cdot z+t^2\|z\|^2.
\]
Dividing by $t$ and letting $t$ decrease to zero yields $y\cdot z\ge0$. Thus $K\subseteq K^*$.

Conversely, take $x\in K^*$ and put $u=g(x)$ and $v=g(-x)$. Both vectors belong to $K$, and Step 1 gives $x=u-v$. Theorem~\ref{thm:sphere}, applied to the profile $(0,x)$ with output $u$, gives $u\cdot(u-x)=0$, or $u\cdot v=0$. It follows that
\[
0\le x\cdot v=(u-v)\cdot v=-\|v\|^2.
\]
Hence $v=0$ and $x=u\in K$. We have proved $K=K^*$. In particular, $K$ is closed and convex, because $K^*$ is an intersection of closed half-planes through the origin.

\textbf{Step 3: Identify the planar cone.}

The cone $K$ contains no nonzero vector together with its negative, since their inner product would be negative. It also spans the plane. Otherwise, there would be a nonzero vector $z$ orthogonal to its span, and both $z$ and $-z$ would belong to $K^*=K$, a contradiction. Thus $K$ is a closed planar sector with two boundary rays. Write their angles as $\alpha$ and $\beta$, choosing $\alpha<\beta$ and $0<\beta-\alpha<\pi$. The nonnegative inner products from Step 2 imply $\beta-\alpha\le\frac{\pi}{2}$.

A vector has a nonnegative inner product with every vector of this sector exactly when its angle lies between $\beta-\frac{\pi}{2}$ and $\alpha+\frac{\pi}{2}$. These are therefore the boundary angles of $K^*$. Since $K=K^*$, the two sectors agree, and $\beta-\alpha=\frac{\pi}{2}$. Consequently, there are orthogonal unit vectors $e,h$ such that
\[
K=\{se+th:s,t\ge0\}.
\]

\textbf{Step 4: Recover the mechanism.}

Step 1 identifies $g(x)$ as a closest point in this quadrant. Writing $x=(e\cdot x)e+(h\cdot x)h$, the squared distance to $se+th$ is $(e\cdot x-s)^2+(h\cdot x-t)^2$. Its unique minimum over $s,t\ge0$ is attained at $s=\max\{e\cdot x,0\}$ and $t=\max\{h\cdot x,0\}$. Translation invariance now yields
\[
f(A,B)=e\max\{e\cdot A,e\cdot B\}
       +h\max\{h\cdot A,h\cdot B\}.
\]
Thus every mechanism under consideration takes coordinatewise maxima in some fixed orthonormal basis. Changing the sign of a basis vector expresses a maximum as a minimum in the oppositely oriented coordinate.

For completeness, we verify that Mechanisms~\ref{mech:c1}--\ref{mech:c3} cover all such bases. If a boundary ray of $K$ is a coordinate axis, its perpendicular boundary is also a coordinate axis, giving Mechanism~\ref{mech:c1}. Otherwise, the interior of the right-angle sector contains exactly one of the four coordinate-axis rays. If it contains the positive vertical ray, its boundary with positive horizontal coordinate is proportional to $(1,u)$ for some $u>0$, and its other boundary is proportional to $(-u,1)$. This is Mechanism~\ref{mech:c2} with parameter $u$. If it contains the negative vertical ray, the boundary with positive horizontal coordinate is proportional to $(1,u)$ for some $u<0$, and the other is proportional to $(u,-1)$. This is again Mechanism~\ref{mech:c2}, now with negative parameter: it takes the maximum in the $e$-coordinate and the minimum in the $h$-coordinate used in the sufficiency proof. The sectors containing the positive or negative horizontal ray are obtained by exchanging the original coordinates, giving Mechanism~\ref{mech:c3} with positive or negative parameter, respectively. This exhausts all sectors and proves necessity.
\end{proof}

\section{Multiagent Outputs and Maximum Cost}\label{sec:discussion}

\subsection{A Limitation of the Orthogonality Condition}

A natural extension of the two-agent sphere condition would require that every output $W$ satisfy
\begin{align}\label{eq:multiagent-orthogonality}
(A_i-W)\cdot(A_j-W)=0\quad\text{for some }i,j\in N.
\end{align}
The indices may coincide, so an output at a reported location satisfies this condition. We show that \eqref{eq:multiagent-orthogonality} need not hold in dimensions at least three, even for the coordinatewise median.

\begin{mechanism}[Coordinatewise Median]\label{mech:median}
In every coordinate, select the median of the reported coordinates. When $n$ is even, select the larger of the two middle values.
\end{mechanism}

This rule is deterministic, unanimous, anonymous, translation-invariant, and strategyproof. For the last property, fix the other reports and consider one coordinate. As the remaining report varies, the chosen order statistic is that report clipped to an interval determined by the other reports, allowing an unbounded endpoint. Truthful reporting selects a closest attainable coordinate to the true coordinate. Thus a deviation cannot reduce the absolute error in any coordinate, and cannot reduce the $L_p$ distance.

\begin{lemma}\label{lem:median-dimensions}
For $m\le2$, the coordinatewise median satisfies \eqref{eq:multiagent-orthogonality} for every number $n\ge2$ and every profile. For each $m\ge3$, there is a three-agent profile for which it fails.
\end{lemma}

\begin{proof}
In one dimension the output is a report, so the condition holds with $i=j$. In two dimensions, translate the output to zero. If some report is zero, the same argument applies. Otherwise, the chosen first coordinate is attained by some agent $s$, and the chosen second coordinate by some agent $t$. Hence $A_s=(0,a)$ and $A_t=(b,0)$ for nonzero $a,b$. These are distinct agents and their vectors are orthogonal.

For $m\ge3$, take the reports $A_1=(0,1,-1,0,\ldots,0)$, $A_2=(-1,0,1,0,\ldots,0)$, and $A_3=(1,-1,0,0,\ldots,0)$. Their coordinatewise median is zero. Each vector has squared Euclidean norm 2, while every pair of distinct vectors has inner product $-1$. Thus no choice of $i,j$ satisfies \eqref{eq:multiagent-orthogonality}.
\end{proof}

The failure concerns a universal statement over numbers of agents and profiles. It does not mean that every fixed number of agents admits such a counterexample: for two agents, the coordinatewise median satisfies the orthogonality condition in every dimension.

\subsection{A Tight Maximum-Cost Bound}

As a final application of strategyproofness, we give a short proof of the deterministic maximum-cost bound. The ratio 2 is already known on the line \cite{Procaccia2009}; the argument below works in every finite-dimensional normed space and for each fixed number of agents.

\begin{lemma}\label{lem:max-cost}
For every $n\ge2$ and $m\ge1$, the optimal approximation ratio of a deterministic strategyproof mechanism for maximum cost in $(\mathbb R^m,\|\cdot\|_p)$ is 2, for every $1\le p\le\infty$.
\end{lemma}

\begin{proof}
Fix any number $n\ge2$ of agents and two distinct points $A,B$. Consider the profile with $n-1$ agents at $A$ and one at $B$, and write $W$ for its output. If $W=B$, the mechanism's maximum cost is $d(A,B)$, whereas the optimal maximum cost is $\frac{d(A,B)}{2}$.

If $W\ne B$, move the agents at $A$ to $W$, one at a time. Each move preserves the output: an agent whose true location is $W$ could otherwise report her previous location and obtain cost zero. The resulting profile has $n-1$ agents at $W$ and one at $B$, and its output is $W$. Its maximum cost is $d(W,B)$, while its optimum is $\frac{d(W,B)}{2}$. Thus the approximation ratio is at least 2 in either case.

The bound is tight: a dictatorial mechanism has maximum cost at most twice the optimum, by the triangle inequality. This argument holds in every $L_p$ norm, including $p=1$ and $p=\infty$.
\end{proof}

\section{Conclusion}\label{sec:conclusion}

We establish necessary output equations and regularity properties for deterministic strategyproof facility location under translation invariance. The equations hold for two agents throughout $1<p<\infty$, while scalability holds for any number of agents. In Euclidean space, the resulting orthogonality condition leads to complete characterizations under rotation invariance or, in two dimensions, anonymity. The general $L_p$ equations remain necessary conditions rather than a complete characterization. Understanding which additional restrictions make them sufficient, and identifying suitable geometric conditions for more than two agents, are natural directions for further work.

\bibliographystyle{plain}
\bibliography{reference}

\end{document}